\pdfoutput=1
\documentclass[conference]{IEEEtran}
\IEEEoverridecommandlockouts                           
\renewcommand{\baselinestretch}{0.95}
\usepackage[T1]{fontenc}

\usepackage{mathtools}
\usepackage{comment}	
\usepackage{enumerate}
\usepackage{makecell}
\setcellgapes{.028in}

\usepackage{epsfig,graphicx,subcaption,color,url} 
\usepackage{amssymb,amsmath,algorithm}
\usepackage{amsfonts,amsthm, bm, bbm,dsfont}
\let\underbrace\LaTeXunderbrace
\usepackage{algorithm}
\usepackage{times}
\usepackage{graphics}
\usepackage{epsfig}
\usepackage{wrapfig}
\usepackage{cite}
\usepackage{amsmath}
\usepackage{algpascal}
\usepackage{algc}
\usepackage{algcompatible}
\usepackage{multicol}
\usepackage{algpseudocode}
\usepackage{bm}
\usepackage{bbm}
\usepackage{caption}
\usepackage{subcaption}
\usepackage{wrapfig}

\usepackage{mathabx}
\usepackage{calrsfs}
\DeclareMathAlphabet{\pazocal}{OMS}{zplm}{m}{n}
\usepackage{url}
\usepackage{comment}

\newtheorem{thrm}{Theorem}
\newtheorem{rem}{Remark}

\renewcommand{\a}{\bm{a}}
\newcommand{\e}{\bm{e}}
\renewcommand{\u}{\bm{u}}
\newcommand{\n}{\bm{n}}
\newcommand{\zero}{\bm{0}}
\newcommand{\indic}{\mathds{1}}
\renewcommand{\S}{\pazocal{S}}
\newcommand{\F}{\pazocal{F}}

\newcommand{\A}{\pazocal{A}}

\newcommand{\bR}{\mathbb{R}}
\newcommand{\bP}{\mathbb{P}}

\renewcommand{\b}{{\bf{b}}}

\newcommand{\bE}{\mathbb{E}}

\newcommand{\D}{\pazocal{D}}

\renewcommand{\it} {{{i,t}}}

\newcommand{\tR}{{\tilde{R}}}
\newcommand{\baR}{{\bar{R}}}

\newcommand{\T}{\pazocal{T}}
\newcommand{\ts}{\tilde{s}}
\newcommand{\ns}{\tilde{\n}}
\newcommand{\tV}{\widetilde{V}}
\newcommand{\tm}{\tilde{m}}

\usepackage{xcolor}

\def\BibTeX{{\rm B\kern-.05em{\sc i\kern-.025em b}\kern-.08em
    T\kern-.1667em\lower.7ex\hbox{E}\kern-.125emX}}

\begin{document}

\title{Post-Decision State-Based Online Learning for\\ Delay-Energy-Aware Flow Allocation in Wireless Systems\\
\thanks{M. G. Bhat and S. Moothedath are with Electrical and Computer Engineering, Iowa State University. Email: \{mgbhat, mshana\}@iastate.edu. P. Chaporkar is with Electrical Engineering, Indian Institute of Technology Bombay. Email: chaporkar@ee.iitb.ac.in. This work is supported by NSF-CNS 2415213, USA and MeitY, India.}}

\author{Mahesh Ganesh Bhat, Shana Moothedath,~\IEEEmembership{Senior Member,~IEEE,} and Prasanna Chaporkar}

\maketitle
\begin{abstract}
We develop a structure-aware reinforcement learning (RL) approach for delay- and energy-aware flow allocation in 5G User Plane Functions (UPFs). We consider a dynamic system with $K$ heterogeneous UPFs of varying capacities that handle stochastic arrivals of $M$ flow types, each with distinct rate requirements. 
We model the system as a Markov decision process (MDP) to capture the stochastic nature of flow arrivals and departures (possibly unknown), as well as the impact of flow allocation in the system. To solve this problem, we propose a post-decision state (PDS) based value iteration algorithm that exploits the underlying structure of the MDP. By separating action-controlled dynamics from exogenous factors, PDS enables faster convergence and efficient adaptive flow allocation, even in the absence of statistical knowledge about exogenous variables.
Simulation results demonstrate that the proposed method converges faster and achieves lower long-term cost than standard Q-learning, highlighting the effectiveness of PDS-based RL for resource allocation in wireless networks.
\end{abstract}
 \begin{IEEEkeywords}
  Resource allocation, UPF flow allocation, Reinforcement learning, Energy-aware decision-making.
 \end{IEEEkeywords}

\section{Introduction}
The evolution of modern 5G and beyond networks (B5G) has enabled a wide range of applications with diverse service requirements, including high-throughput broadband, latency-critical control applications, and large-scale IoT connectivity. This has led to increasing network traffic with heterogeneous and stringent quality of service (QoS) requirements. The rapid growth of infrastructure to support large amounts of traffic has led to an increase in the energy footprint. As a result, sustainability is a primary objective in the design of modern networks. 
The recommendation ITU-R M.2160 identified sustainability as one of the highlights in the clauses of ``\textit{Motivation and societal considerations}" and ``\textit{User and application trends}".
This motivates energy-aware resource management and greener network deployments. With this objective in focus, while extensive work has addressed resource allocation and scheduling in Radio Access Networks (RAN)~\cite{11160834, 11161043}, it is important to consider the same in the User Plane Function (UPF), which plays a pivotal role in the 5G core.

The UPF acts as a central data-plane interconnection point within the 5G core, responsible for data forwarding, traffic routing, and enforcing QoS policies. Traffic flow allocation to UPFs directly affects network performance and reliability. Furthermore, UPFs are virtual functions deployed on a variety of infrastructures and vary in characteristics such as geographical location, computational capacity, user proximity, and energy profiles (the amount of operational energy contribution from green energy). These factors influence the selection of the UPF and underscore the need for a principled approach for flow allocation, rather than a simple data forwarding approach.

Designing intelligent and adaptive flow allocation strategies is essential for achieving scalable, delay-sensitive, and resource-efficient performance in next-generation mobile networks \cite{wang2024delay}.
The flow allocation problem shares characteristics with the task scheduling problem and has been extensively studied in computing and communication systems. Earlier works primarily addressed static (optimization version) allocation problems \cite{BRAUN20081504, han2008resource}, limiting adaptability in dynamic and uncertain environments.
%
%
Reinforcement learning (RL) presents a promising framework for addressing the challenge of flow allocation by enabling data-driven, adaptive decision making in dynamic and uncertain network environments \cite{nasir2019multi, tang2020deep}. Unlike traditional methods, RL can learn optimal policies through interaction with the system, accounting for stochastic arrivals, departures, and system constraints, making it particularly well-suited for modern, delay-sensitive and resource-constrained applications in next-generation networks.
For instance, \cite{9078843} proposed a hybrid approach using RL and stochastic gradient descent for Mobile Edge Computing (MEC) offload scheduling with random task arrivals. 
However, it considered only a single edge cloud.
A Deep Reinforcement Learning (DRL) approach is proposed for delay-resource-aware service optimization in satellite-deployed UPFs \cite{wang2024delay} but they consider a fixed set of arrivals.
DRL based resource allocation schemes are proposed for minimizing the delay in MEC in~\cite{8657791,8798668}. DRL-based approaches face two major limitations: they learn approximate rather than optimal policies and lack formal guarantees. Additionally, they depend on function approximations and typically require large amounts of training data. Further, in all of these studies, the task arrivals and departures are assumed to be fixed and known. However, the task arrivals and departures are often exogenous and unknown. 
Our goal in this work is to leverage the problem structure and develop efficient and optimal RL algorithm that offers fast and sample efficient learning.

The inherent structural properties of allocation and scheduling problems enable post-decision state (PDS)-based RL analysis. PDS analysis has been applied in scheduling and allocation problems. PDS has been studied in the context of delay-sensitive wireless transmission scheduling for energy-efficient point-to-point communication and accelerated learning~\cite{5986747, 6093953}. A PDS-based online learning approach for server allocation in data centers was introduced in \cite{7491319}, aiming to reduce electricity costs. A wireless resource scheduling in virtualised RAN networks with random arrivals and departures focusing on mobile device utility is considered in \cite{8014506}, where a multi-agent MDP is decomposed into single-agent MDPs, and a PDS-based online localized algorithm is proposed. RL and DRL-based privacy-aware offloading methods for IoT applications using PDS are studied in~\cite{8515032, 8491311}.

While several works address the flow scheduling problem in the context of wireless mobile network, they rely on the extensive training data, known system dynamics or fixed traffic flows. In contrast to these works, we focus on modeling the multi-flow dynamic allocation, under unknown flow arrivals and departures. This paper develops an RL-based framework for dynamic flow allocation in 5G networks with multiple UPFs and heterogeneous traffic classes. By formulating the problem as a Markov Decision Process (MDP), we capture the stochastic nature of the network traffic and propose a novel RL algorithm to learn optimal allocation policies when transition probabilities of MDP are not known. We aim to leverage the structural properties to learn the optimal policy using PDS-based RL to solve delay and energy-aware flow allocation, while considering the capacity constraints. The proposed approach enables real-time, adaptive decision-making, addressing the limitations of static or model-based techniques in highly dynamic environments to optimize performance. To the best of our knowledge, this problem in a multi-server scenario with stochastic (unknown) arrivals and departures has not been previously studied.
%

The key contributions of this work are threefold.
\begin{enumerate}
    \item[$\bullet$] We formulate the multi-flow dynamic allocation problem as a delay- and energy-aware MDP, capturing the stochastic traffic arrivals and departures under fixed capacity constraints to enable adaptive and principled allocation.
    \item[$\bullet$] We propose a post-decision state (PDS)-based value iteration (VI) algorithm that leverages the problem structure, specifically, the separation between exogenous and controllable dynamics, to simplify learning, enhance sample efficiency, and accelerate convergence. We provide a convergence guarantee of the proposed PDS-VI algorithm.
    \item[$\bullet$] We evaluated the performance of the proposed approach via numerical simulations and demonstrate faster convergence and efficient cost performance compared against a conventional baseline.
\end{enumerate}

\section{Preliminaries: Markov Decision Process}\label{sec:not}
%
Markov Decision Process (MDP) is the standard framework for modeling a
stochastic dynamical system and  computing its optimal control policy~\cite{sutton2018reinforcement}. 
 Let $\S$ and $\A$ be compact  sets describing the states and actions of the controller (agent), respectively, and $r:\S \times \A \rightarrow \bR$ be the reward function.  The system dynamics is characterized by the probability transition structure $\bP$, where $\bP(s'|s, a)$ is the probability of transitioning to state $s'$ from state $s$ under control action $a$. A  policy $\pi: \S \rightarrow \A$ is a conditional distribution $\pi(a|s)$ that guides the decision-making process of an agent. At time $t \in \{0,1,\ldots,\}$,  the agent  observes the current state $s_t$ and chooses an action $a^{\pi}_t$ from the policy $\pi(a|s)$, and observes the reward $r(s_t,  a^{\pi}_t)$.  The action chosen by the agent at a state drives the agent to a next state with probability (w.p.) $\bP(s_{t+1}|s_t, a^{\pi}_t)$. 
 The agent's goal is to find an optimal policy $\pi^\star$ that maximizes the cumulative reward $\lim_{T \rightarrow \infty} \sum_{t=0}^{T}\gamma^t r(s_t, a^{\pi}_t)$, where $\gamma \in (0, 1]$ is the discount factor that captures how myopic the agent is.

\section{Problem Formulation: Delay-Sensitive Flow Allocation in 5G User Plane Functions}

\subsection{System Model}
Consider a 5G network. The time is slotted. 
At the beginning of each slot, a new flow arrives in the network w.p.~$p$.
Each flow arriving can be of one of $M$ types. Let $\b_m$ denote the probability that an arriving flow is of type $m$ for $m\in [M]$, where $[Z]:=\{1,2,\ldots, Z\}$ for any integer $Z$.
We also assume that the flow arrivals and its type are independent across flows.
The flow type indicates the average flow rate requirement.
Let $\baR_m$ denote the average rate requirement for the flow of type $m$. 
Without loss of generality, 
$\baR_m < \baR_{m+1}$ for every $m \in [M-1]$.
For each flow that arrives, the network must decide whether to accept the flow. Specifically, if the required rate cannot be guaranteed, then the admission is denied to the flow; otherwise it is accepted. 

If a flow is accepted in the network, then it is assigned a UPF that handles the flow until it departs and provides it the required rate based on its type. 
We assume that there are $K$ UPFs in the network. Each UPF may have distinct capabilities in terms of the available memory, computational and switching speeds. Based on the capabilities, let $C_k$ denote the maximum data rate (in bits/second) that UPF $k \in [K]$ can support.
Finally, at most one flow departs the $k^{\rm th}$ UPF w.p. $q_k$, for $k \in [K]$,  at the end of the slot, independently of flow departures in other slots.
The departing flow is equally likely to be any existing flow in the UPF.
A flow arriving in a slot can depart in the same slot.
We now formulate the dynamic UPF allocation problem as an MDP. 

\subsection{Markov Decision Process Formulation}
Flow allocation in 5G networks is inherently a dynamic decision-making problem. MDP offers a principled framework for capturing such dynamics and optimizing long-term performance. 
In this section, we model the delay- and energy-sensitive flow allocation problem in UPFs as an MDP and describe its key components in detail.

\subsubsection{State Space}
We define the state of the system as a tuple containing the allocation matrix $\n$ and the flow arrival state $f$, i.e., $s=(\n, f)$. 
For a state $s=(\n, f)$, let $\n$ be a $K\times M$ matrix such that the $(k,m)^{\rm th}$ entry $\n_{km}$ 
denotes the number of type $m$ flows handled by $k^{\rm th}$ UPF at $s$ and $f \in \{0,1,,\ldots, M\}$ denotes the type of flow arrived at $s$. Here, $f=0$ means there is no flow arrival at $s$. 
Let $\tR_k(s)$ denote the total average rate that 
UPF $k$ needs to support in $s$.
Note that
\begin{align}
    \tR_k(s) = \sum_{m\in [M]} \n_{km} \baR_m.
\end{align}
Let us define the state space $\S$ as
\begin{align*}
   \S &= \left\{s: C_k > \tR_k(s) \mbox{~for~all~$k \in [K]$}, f\in \{0,\ldots, M\}\right\}.
\end{align*}
Note that $\lfloor C_k/\baR_1\rfloor$ is the maximum number of flows that UPF $k$ can support at any
given time.
Thus, $\S$ is a finite set. 

\subsubsection{Action Space and Control Policy}
    Consider a state $s=(\n,f)$. 
    Let $\A(s)$ be the set of feasible actions in $s$. 
    Then $\A(s)\subseteq \{0,1,\ldots, K\}$, where $k\in \A(s)$ if $C_k > \tR_k(s) + \baR_f$, for $k \in [K]$.
    For $a\in\A(s)$, $a\in[K]$ indicates flow $f$ is allocated to UPF $a$ and $a=0$ indicates the flow is blocked.
The flow $f$ in $s$ must be admitted as long as $\A(s)\neq \{0\}$ and will be handled by the allocated UPF until it departs.
The action set is of dimension $K+1$. For $a\in\A(s)$, we define a $K \times M$ indicator matrix, $\a$, for analytical purposes. 
\begin{align}\label{eq:ind}
\a_{km}=
\begin{cases}
1, & \text{if } a=k \text{ and }f = m,\\
0, & \text{otherwise}.
\end{cases}
\end{align} 
Thus, $\a$ is a sparse indicator matrix with at most one non-zero entry (equal to $1$), with all remaining entries being $0$s.

    A control policy $\pi: \S \rightarrow \A$ maps the states to feasible action. 
We assume that $\pi$ is causal, i.e., the action chosen for state $s$ may depend on the past states and actions taken, but not on future evolution. Moreover, $\pi$ can also be a randomized policy, i.e., for a state action can be chosen randomly from the set of feasible actions.
\subsubsection{Cost Function}
Let $\xi(s)$ denote the total cost incurred in state $s$.
We consider 
\begin{align}
\vspace{-2 mm}
    \xi(s) = \sum_{k\in [K]} (\alpha_k(s) + \delta_k(s)),
\end{align}
where $\alpha_k(s)$ denotes power cost, and $\delta_k(s)$ denotes the delay cost at UPF $k$ in 
state $s$. These costs are given as
\begin{align}
\vspace{-2 mm}
    \alpha_k(s) &= c_k \tR_k(s) \mbox{~and~} \delta_k(s) = \frac{C_k}{C_k - \tR_k(s)}.
\end{align}
Here, $c_k$ is power cost for switching one bit at UPF $k$. 

\subsubsection{System Dynamics}
The system dynamics is characterized by the probability transition matrix $\bP$, where $\bP(s'|s, \a)$ is the probability of transitioning to the state $s'$ from the state $s$ under control action $\a$. Consider the current state $s=(\n, f)$ and the next state $s'=(\n', f')$, where 
$\n, \n'$ denote the flow allocation matrices and $f, f'$ the flow arrivals. 

The transition dynamics consist of two components: (i) the evolution of the allocation matrix $\n$, which depends on the current state, control action, and stochastic departures, and (ii) the arrival process $f$, which is exogenous and independent of the control action, follows a fixed but unknown probability distribution.
Formally, the transition probability 
\begin{align}
    \bP(s'|s, \a) = \bP(\n'|s, \a) \cdot \bP(f'), \label{eq:prob}
\end{align}
where the arrival probability is given by
\vspace{-2 mm}
\begin{align}
    \bP(f') = (1-p)\indic_{\{f'=0\}} + p\sum_{m=1}^M b_m \indic_{\{f'=m\}}.\label{eq:arr}
\end{align}
The evolution of the allocation matrix $\n$ depends on the departure process. To define $\bP(\n'|s, \a)$, we first define the departure process. Let $\n_k$ and $\n'_{k}$ be the row vectors corresponding to UPF $k$ in the  allocation matrices at the current state $s$ and next state $s'$, respectively. Define the canonical vector in $\mathbb{R}^M$ as
\[
\vspace{-2 mm}
\e_m \in \{0,1\}^M,\qquad
(\e_m)_j =
\begin{cases}
1, & \text{if } j=m,\\
0, & \text{if } j\neq m
\end{cases}
\]
Consider a scenario where a flow of type $m$ departs.
Let $\D_k(\n'_k, \n_k, f)$ be the transition probability of UPF $k$, i.e., probability of transitioning from $\n_k$ to $\n'_k$ under arrival $f$. Let us first consider the case where no new flow arrives in $s$ or flow of type $\tm$ arrives but $\A(s)$ is empty. Then, either $\n=\n'$, or the allocation matrix transitions to a distinct state due to flow departures. Then,
\begin{align}\label{eq:dep}
\vspace{-2 mm}
    \D_k(\n'_k, \n_k, f) &= \begin{cases}
    1-q_k, & \text{if } \n'_k = \n_k \\
    q_k \cdot \frac{\n_{km}}{\sum_{m'}\n_{km'}}, & \text{if } \n'_{k} = \n_{k} - \e_m \\
    1, & \text{if } \n'_k = \n_k = \zero \\
    0, & \text{otherwise.}
    \end{cases}
\end{align}
Now consider a case where a new arrival of type $\tm$ arrives at state $s$ and $\A(s)$ is non-empty. Without loss of generality, let $a=k'$. Two cases arise.
\textbf{Case 1:} For any UPF $k \ne k'$. The departure probability in this case is same as in Eq.~\eqref{eq:dep}.
\textbf{Case 2:} For UPF $k'$. Then,
\begin{align*}
\vspace{-2 mm}
    &\D_{k'}(\n'_{k'}, \n_{k'}, f) = \hspace{-3 mm}
    &\begin{cases}
    1-q_{k'}, & \text{if } \n'_{k'} = \n_{k'} + \e_{\tm} \\
    q_{k'} \cdot \frac{\n_{k'\tm}+1}{\sum_{m'}\n_{k'm'}+1}, & \text{if } \n'_{k'} = \n_{k'} \\
    q_{k'} \cdot \frac{\n_{k'm}}{\sum_{m'}\n_{k'm'}+1}, & \text{if } \n'_{k'} = \n_{k'} + \e_{\tm}\\&- \e_{m}~\text{and } m \neq \tm\\
    0, & \text{otherwise.}
    \end{cases}
\end{align*}
The transition probability
 \begin{align}
 \vspace{-2 mm}
 \bP(\n'|s,\a) &= \prod_{k=1}^K \pazocal{D}_k(\n_k, \n'_k,f). \label{eq:mat-dynamics}
 \end{align}
 %
Eqs.~\eqref{eq:prob}, \eqref{eq:arr}, \eqref{eq:mat-dynamics} complete the modeling of system dynamics.

\subsubsection{Optimal Policy}
Let $s^{\pi}(t)$ be the state in time slot $t$ under policy $\pi$. For any function $g: \S \to \mathbb{R}$, we use the shorthand $g^\pi(t)$ to denote $g(s^\pi(t))$.
For instance, $\tR_k^\pi(t)$ denotes the total rate that UPF $k$ must support at time $t$ under policy $\pi$, i.e., $\tR_k(s^\pi(t))$, and $\xi^\pi(t)$ denotes the cost incurred at time $t$, i.e., $\xi(s^\pi(t))$.
Our goal is to learn optimal policy $\pi^\star$ that minimizes the cumulative discounted reward
\[
\vspace{-2 mm}
\xi^\pi = \lim_{T \to \infty} \sum_{t=0}^{T} \gamma^t \, \xi^\pi(t).
\]

\section{Proposed Structure Leveraging RL Approach}
In practical settings, the arrival probabilities ($p$, $\b_m$) and the departure probabilities ($q_k$) are often unknown, rendering classical dynamic programming approaches inapplicable. To address this, we propose a model-free RL algorithm.

\subsection{Post-Decision State Based Learning}
Consider a state transition from current state $s=(\n, f)$ and the next state $s'=(\n', f')$.
We know from Eq.~\eqref{eq:prob} that the transition dynamics has two components. This decomposition facilitates the design of RL algorithms that can effectively handle partial knowledge of the environment. 

Let $\ts= (\ns, f)$ be a virtual state immediately after taking an action, but before the impact of the stochastic arrivals and departures, which is referred to as the post-decision state.  Here $\tilde{\n}$ is the corresponding allocation matrix. Given this virtual state $\ts$, we can rewrite Eq.~\eqref{eq:mat-dynamics} as
\begin{align*}
\vspace{-2 mm}
    \bP(\n'|s,\a) = \bP(\n'|\ns)\cdot \bP (\ns|s,\a).
\end{align*}
This enables us to decompose the system dynamics in Eq.~\eqref{eq:prob} into action-controlled (known), $\bP^k(\cdot|\cdot, \cdot)$, and stochastic (unknown), $\bP^u(\cdot|\cdot, \cdot)$, transitions. 
\begin{align}
\vspace{-2 mm}
    \bP(s'|s,\a) &= \underbrace{\bP(\tilde{\n}|s,\a)}_{\substack{{\small\mbox{~action~controlled}}\\\bP^k(\ts|s,\a)}} \cdot \underbrace{\bP(\n'|\tilde{\n}) \cdot \bP(f')}_{\substack{\small\mbox{exogenous}\\\bP^u(s'|\ts,\a)}}.\label{eq:pds-dynamics}
    \end{align}

The purely action controlled evolution of the allocation matrix is deterministic, i.e., for $\ts= (\ns, f)$,
\begin{align}\label{eq:deterministic}
\vspace{-2 mm}
\bP^k(\ts|s,\a) =\bP(\tilde{\n}|s,\a) = 1.
\end{align}

The uncertainty in system dynamics arises primarily from the stochastic arrival and departure processes, enabling a structured yet flexible modeling approach for learning-based control.
We can leverage these structural attributes to utilize the potential of post-decision state analysis to reduce complexity \cite{SBKB08, ZHD23}. 
In the next section, we present a reinforcement learning algorithm based on PDS to efficiently learn the optimal policy under partially known system dynamics.
%
\subsection{Proposed Algorithm: PDS-Based Value Iteration}
Let $s=(\n,f)$ be the state of the system in some time slot. Let $a \in \A(s)$ be the action chosen at state $s$ and $\a$ be the corresponding matrix. Then the post-decision state represented by $\ts$ is given by $\ts = (\ns,f) = (\n+\a,f)$. Let $\u$ be the $K\times M$ matrix of departures observed at the end of the time slot whose row vectors, $\u_k$, $k \in [K]$, are given by
\begin{align}
\vspace{-2 mm}
    \u_{k}=\begin{cases}
        \e_m,& \text{flow of type $m$ departs from UPF $k$}\\
        \zero, & \text{if no departures.}
    \end{cases}\label{line:um}
\end{align}
The next actual state, considering the stochastic arrivals and departures, is termed as the \textit{pre-decision state}, $s' = (\n',f')$, where $\n'= \ns-\u = \n+\a-\u$. The sequence of state transitions considering PDS is shown in Figure~\ref{fig:VI-seq}.

The Bellman equation for our MDP can be written as,
\begin{align}\label{eq:bellman}
\vspace{-2 mm}
    V(s) &= \min_{a \in \A(s)} \bigg \{ \xi(s) + \gamma \sum_{s'} \bP (s'|s,\a) \cdot V(s')\bigg \}.
\end{align}
By leveraging the PDS transition structure in Eqs.~\eqref{eq:pds-dynamics},~\eqref{eq:deterministic} and substituting for $s'$, we can rewrite Eq.~\eqref{eq:bellman} as
\begin{align*}
\vspace{-2 mm}
    V(s) = \min_{a \in \A(s)} \bigg \{ \xi(s) + \gamma \sum_{f',\u} \bP(\n'|\tilde{\n}) \cdot \bP(f') \cdot V(\ns-\u,f')\bigg \}.
\end{align*}
Now we define the post-decision state value function, $\tV:\S \rightarrow \bR$ as the expected value over all pre-decision states, $s'$, reachable from the post-decision state,
\begin{align}
\vspace{-2 mm}
    \tV(\ts) &= \bE_{s'}[V(s')] \nonumber\\
    \tV(\ts) &= \sum_{s'} \bP(\n'|\ns)\cdot \bP(f') \cdot V(s'). \label{eq:post-state}
\end{align}
With this definition, the Bellman equation becomes
\begin{align}
\vspace{-2 mm}
    V(s) = \min_{a \in \A(s)} \Big\{\xi(s) + \gamma \tV(\ts)\Big\}.\label{eq:bellman_2}
\end{align}
  \begin{figure}[t]
      \centering      \includegraphics[width=0.48\textwidth]{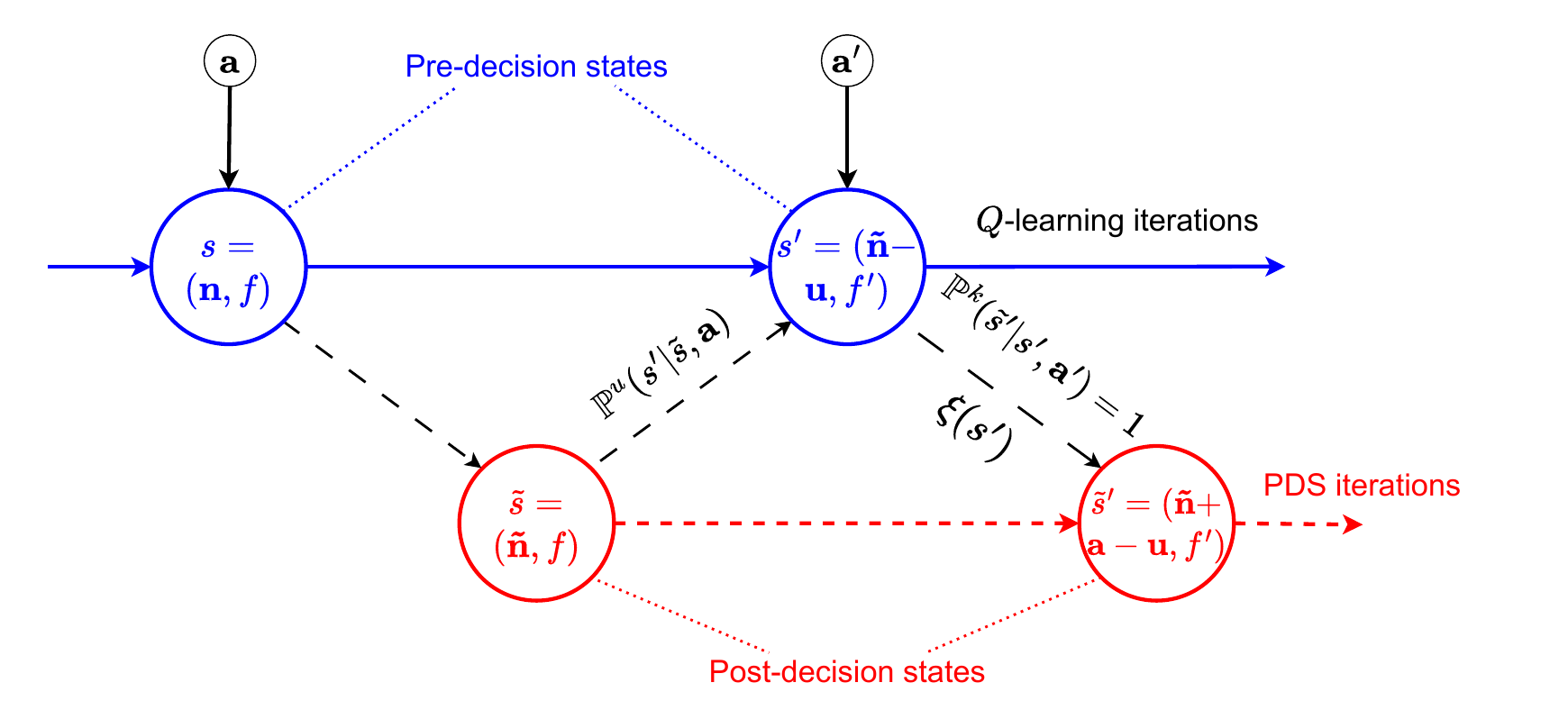}
      \caption{Illustration of PDS-based state transition.}
      \label{fig:VI-seq}
  \end{figure}
Let $a' \in \A(s')$ be the action chosen at state $s'$ and $\a'$ be the corresponding indicator matrix. Also, let the next post decision state after selecting $a'$ at $s'$ be $\ts' = (\n'+\a',f')$. 
We can write the value equation for state, $s'$ as 
 \begin{align*}
 \vspace{-2 mm}
     V(s') = \min_{a' \in \A(s')} \bigg \{ \xi(\n',f') + \gamma \tV(\n'+\a',f') \bigg \}.
 \end{align*}
substituting this in Eq.~\eqref{eq:post-state}, we obtain the equation that forms the basis for the proposed value iteration algorithm,
\begin{multline} \label{eq:2}
\vspace{-2 mm}
\tV(\ts) = \sum_{f',\u} \bP(\n'|\tilde{\n}) \cdot \bP(f') \min_{a' \in \A(s')} 
    \Big[ \xi(\ns-\u,f') +\\
     \gamma \tV(\ns-\u+\a',f')\Big].
\end{multline}

\begin{rem}
    Compared to the standard Bellman equation, using the structural property of the formulation, the \textit{expectation} is outside the \textit{minimization} in Eq.~\eqref{eq:2}. This enables us to propose a value learning using stochastic approximation.
\end{rem}
\begin{rem}
 Eq.~\eqref{eq:2} is the Bellman's equation for the PDS value function $\tV$. 
 By substituting the optimal value corresponding to the PDS analysis, $\tV^\star$, in Eq.\eqref{eq:bellman_2}, the optimal value vector of Eq.~\eqref{eq:bellman}, $V^\star$, can be obtained. This in turn, this also gives us the optimal policy $\pi^\star$.
\end{rem}
\begin{algorithm}[t]
\caption{Value iteration with post-decision state}
\label{Alg:Value}
\begin{algorithmic}[1]
\Require Number of UPFs $K$, Number of flow types $M$, Resource requirements $\baR_m$, Capacities $C_k$, Unit cost $c_k$
\Ensure Policy $\pi$
\State Initialize post-state $\ts = (\ns,f), \ts \in \S$, $\tV\leftarrow \zero$, and $t = 0$
\For {$t = 0,1,2,3,\ldots$}

    \State Observe departures and compute departure matrix $\u$ using Eq.~\eqref{line:um} \hfill //departures
    \State Observe the flow arrival $f'$ \hfill // arrival
    
    \State Compute $s' = (\ns - \u,f')$    
    \State Determine feasible actions, $\A(s')$
    \State Update $\tV_{t+1}(\ts)$ using Eq.~\eqref{eq:ivi} and obtain $\a'$ corresponding to the minimizing action $a' \in \A(s')$
    \State Compute $\ts' = (\ns - \u + \a', f')$
    \State Update $\pi(s') = a'$ and $\ts = \ts'$
    \EndFor
\end{algorithmic}
\end{algorithm}

Next, we describe how a learning algorithm, the PDS-based value iteration (PDS-VI), can be obtained from Eq.~\eqref{eq:2}. Let $\{\alpha_{t}\}_{t \geqslant 0}$ be a positive step-size sequence satisfying the Robbins-Monro conditions $\sum_t \alpha_{t} = \infty$ and $\sum_t (\alpha_{t})^2 < \infty.$
Using these step sizes, we perform stochastic approximation of value vector and update the estimate of each state at every time step depending on the observed arrivals and departures according to the following 
update rule
\begin{align} \label{eq:ivi}
\tV_{t+1}(\ts) =& \tV_{t}(\ts) + \alpha_{t} \Bigg[
    \min_{a' \in \A(s')} \Big( \xi(\ns-\u,f') + \nonumber\\
    &\gamma \tV_{t}(\ns-\u+\a',f')\Big)-\tV_{t}(\ts)\Bigg],\\
\tV_{t+1}(\ts'') =& \tV_{t}(\ts''), \qquad \mbox{~for~all~} \ts'' \neq \ts. \nonumber
\end{align}


The iterative algorithm to perform value estimation based on the given update is provided in Algorithm~\ref{Alg:Value}. Below we prove that these iterates converge to the optimal value of the PDS value vector, $\tV^\star$.
\begin{thrm}
    The PDS value function iterates in Eq.~\eqref{eq:ivi} converge to the optimal PDS value function, $\tV_{t} \rightarrow \tV^\star$.
\end{thrm}
\begin{proof}
Let $\T:\bR^{|\S|}\rightarrow\bR^{|\S|}$ be defined as,
\begin{multline}
\T(\tV)(\ts) = \sum_{f',\u} \bP(\n'|\tilde{\n}) \cdot \bP(f') \min_{a' \in \A(s')} 
    \Big[ \xi(\ns-\u,f') +\\
     \gamma \tV(\ns-\u+\a',f')\Big].
\end{multline}
Define $\widehat{\T}$ as the sampled operator for some observed $f'$ and $\u$ at time $t$, then 
\begin{align*}
    \widehat{\T}(\tV_{t})(\ts) = \min_{a' \in \A(s')} \Big[ \xi(\ns-\u,f') +   \gamma \tV_{t}(\ns-\u+\a', f') \Big].
\end{align*}
\vspace{-2 mm}
We can rewrite Eq.~\eqref{eq:ivi} in terms of $\T$ and $\widehat{\T}$ as
\begin{align*}
\tV_{t+1}(\ts) = \tV_{t}(\ts) + \alpha_{t} \left( \T(\tV_{t})(\ts) - \tV_{t}(\ts) + \epsilon_{t+1}(\ts) \right),
\end{align*}
where, $\epsilon_{t+1}(\ts) = \widehat{\T}(\tV_t)(\ts) - \T(\tV_t)(\ts)$. Rearranging,
\begin{align}
    \tV_{t+1}(\ts)-\tV_{t}(\ts)=\alpha_{t} \left(\T(\tV_{t})(\ts) - \tV_{t}(\ts) + \epsilon_{t+1}(\ts) \right).\hspace{-2 mm}\label{eq:sa_up_1}
\end{align}
Let $\F_t= \{(s_{t'},a_{t'},\xi_{t'},s'_{t'})\}_{0\leqslant t' \leqslant t}$ represent the information up to time $t$. It follows that $\{\epsilon_{t+1}(\ts)\}_{t \geqslant 0}$ forms a martingale difference sequence and $\bE[\epsilon_{t+1}(\ts)| \F_t] = 0$. The iterates in Eq.~\eqref{eq:sa_up_1} are equivalent to the discretized version of the ordinary differential equation (ODE),
   $\dot{\tV} = \T(\tV) - \tV.$
From \cite{borkar2000ode}, it can be argued that the convergence of these updates is equivalent to the convergence of the aforementioned ODE. Also, $\T$ is a $\gamma$-contraction under the sup-norm i.e.,
\begin{align*}
    ||\T(f) - \T(g)||_{\infty} \leqslant \gamma ||f-g||_{\infty}\quad \forall f,g.
\end{align*}
Thus it is guaranteed that the iterates $\tV_t$ converge to the optimal value $\tV^\star$.
\end{proof}
\begin{rem}
The upper bound for cardinality of $\S$ is given by $|\S|=(M+1)\cdot\prod_{k=1}^{K}(\lfloor C_k/\baR_1\rfloor)^M$ and this grows exponentially with $K$ and $M$. The PDS-VI saves both computations and memory: PDS-VI algorithm needs to store the value estimate at each state and hence has a complexity of $|\S|$ as opposed to Q-learning which requires $|\S| \times |\A|$ to store all state-action pairs. The PDS-VI and Q-learning both have computational complexity of $|\pazocal{A}|$ per iteration, however, the number of iterations for convergence in PDS-VI is in the order of $|\S|$ and in Q-learning is in the order of $|\S|\times|\A|$. 
\end{rem}
\begin{rem}
The application of PDS analysis in our proposed model accelerates convergence by leveraging structural properties beyond the standard PDS decomposition.
 In our model, the PDS value is independent of the current flow. Thus $|M|+1$ states can be simultaneously updated further improving convergence speed. 
\end{rem}
\begin{figure*}[t]
    \centering
    \begin{subfigure}[b]{0.32\textwidth}
        \centering
        \includegraphics[width=\textwidth]{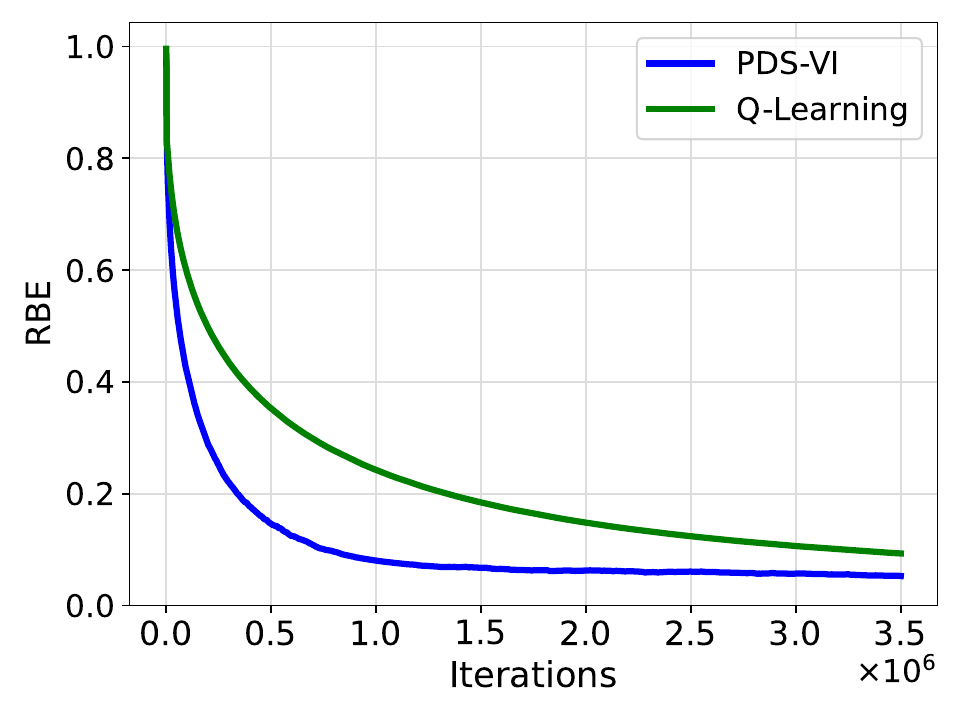}
        \vspace{-6 mm}
        \caption{}
        \label{fig:4upf_rbe_M}
    \end{subfigure}%
    \begin{subfigure}[b]{0.32\textwidth}
        \centering
        \includegraphics[width=\textwidth]{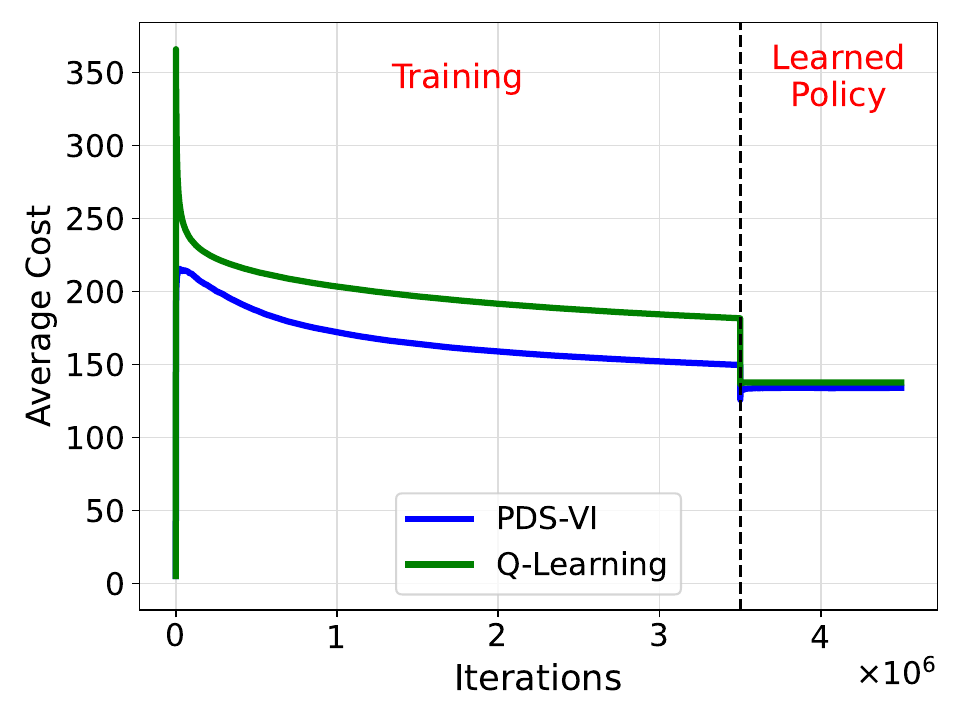}
        \vspace{-6 mm}
        \caption{}
        \label{fig:4upf_total_cost_M}
    \end{subfigure}%
    \begin{subfigure}[b]{0.32\textwidth}
        \centering
        \includegraphics[width=\textwidth]{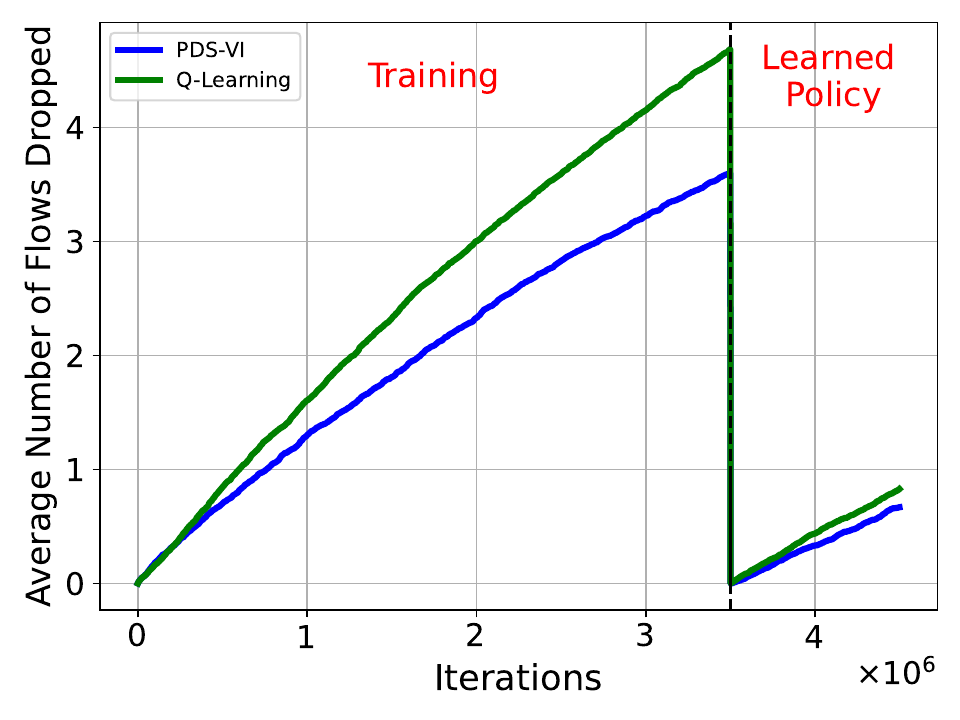}
        \vspace{-6 mm}
        \caption{}
        \label{fig:4upf_flow_drop_M}
    \end{subfigure}
    \vspace{-1 mm}
    \caption{PDS-based value iteration vs. Q-learning. Figure~\ref{fig:4upf_rbe_M} presents plots for average relative Bellman error, RBE, Figure~\ref{fig:4upf_total_cost_M} presents average cost, and Figure~\ref{fig:4upf_flow_drop_M} represents the number of flows blocked with respect to iteration for 5 UPF case.}
    \label{fig:results_plot_2}
\end{figure*}

\section{Numerical Simulations}
%
\noindent{\bf Data Generation:}
We evaluated the proposed algorithm on a synthetic dataset with $5$ UPFs and $2$ flow types. The flow arrivals are sampled from a Bernoulli distribution with probability $p = 0.7$, and the flow departures are dynamically sampled with $q_k = 0.3$, for all $k \in [K]$. Their probabilities are $b_1 = 0.6$ and $b_2=0.4$, respectively. The average rate requirement for each flow type is set as $\bar{R}_1 = 30$ and $\bar{R}_2 = 35$. The maximum data rate $C_k$ is set to $100$ for each UPF. We assume a discount factor of $0.96$. 
The total number of states is $98304$. The power costs are $c_1 = 5, c_2 =4, c_3 = 3, c_4 =2$, and $c_5 = 1$. The simulation is run over $4.5 \times 10^6$ slots with $3.5 \times 10^6$ for training and $1 \times 10^6$ for evaluation.

\noindent{\bf Baseline and Metrics:}
We performed value iteration (dynamic programming) to obtain the ground truth. 
The update equation for value iteration is 
%
    $V_{t+1}(s) = \xi(s) + \min_{\a \in \A} \gamma \sum_{s'} \bP(s'|s,\a) V_{t}(s)\ \ \forall s \in \S.$
The value update leverages knowledge of the transition dynamics and thus serves as the ground-truth baseline for evaluating the performance metrics.
%
We compared the algorithms based on the speed of convergence and the average reward. We consider the time average cost of the algorithm
\vspace{-2 mm}
\begin{align*}
    \bar{\xi}_{t} = \frac{1}{t}\sum_{t=1}^{N} \xi(s_{t}), 
\end{align*}
for cost performance and the relative Bellman error for convergence. Since states are visited at different frequencies, we opt for a weighted error, given by 
\[
    RBE = \frac{\sum_{s}w_{s}(|V(s) - V^{\star}(s)|)}{\sum_{s}w_{s}(|V^{\star}|(s))},
\]
where $w_s$ is number of times the state $s$ is visited and $V^{\star}$ is the ground truth. We also compared the number of flows dropped by each algorithm as an evaluation metric.

\noindent{\bf Results and Discussion:} From Q-learning, we obtain value function using $V(s) = \min_a Q[s,a]$. It serves as the model-free learning baseline.
Our objective is to demonstrate that the proposed PDS approach achieves performance comparable to Q-learning more efficiently, i.e., with fewer data samples and reduced computational time.
We bifurcate the cost evaluation into two sections; the cost incurred during training and a short evaluation of the learned policy post-training. We implemented the proposed algorithm using Python. The simulations were run on a Windows machine with \textit{Intel(R) Xeon(R) W-1370} processor. 
The simulation results in Figure~\ref{fig:results_plot_2} are averaged over 1000 independent Monte Carlo runs.

As shown in Figure~\ref{fig:4upf_rbe_M}, PDS-VI converges much faster, reaching $90\%$ accuracy within $7.5 \times 10^5$ iterations, whereas Q-learning takes about $3.5 \times 10^6$ iterations. 
As expected, due to faster convergence, PDS-VI achieves a lower cost 
during training compared to Q-learning as seen in Figure~\ref{fig:4upf_total_cost_M}. During policy evaluation using the learned policies, both methods perform similarly, with PDS-VI being slightly better consistently. Figure~\ref{fig:4upf_flow_drop_M} shows that PDS-VI blocks fewer flows both in training and evaluation compared to Q-learning. Thus, PDS-VI clearly outperforms Q-learning by converging faster, reducing training cost, and admitting more flows.
%
\begin{rem}
    The Protocol Data Unit (PDU) session anchor UPF, defined in 3GPP~TS~23.501, maintains the PDU session context providing session continuity under mobility while other UPFs serving the flow might change. The Session Management Function (SMF) selects this anchor UPF based on several factors including geographic service area constraints. The UPFs serving the same network are typically deployed in small geographically localized clusters. Hence, the effective selection of anchor UPFs occurs among a small subset of $K$ UPFs, where $K$ is often a small number.
\end{rem}

\vspace{-2 mm}
\section{Conclusion}
In this letter, we studied the delay and energy-aware flow allocation problem in wireless systems under resource constraints. We formulated the problem as an MDP and proposed a model-free RL algorithm based on post-decision state (PDS) learning. By leveraging the decomposable structure of the system dynamics, i.e., separating the controllable and exogenous factors, our approach enables efficient learning with faster convergence. We proved the convergence of the proposed algorithm and evaluated its performance against the standard Q-learning algorithm, validating its effectiveness.

\vspace{-1 mm}
\bibliographystyle{IEEEtran}
\bibliography{main.bbl}
\end{document}